\documentclass[smallextended]{svjour3} 

\smartqed  

\usepackage{amsmath,amssymb}
\usepackage{chngcntr}
\usepackage{mathtools}
\usepackage{graphicx}
\usepackage{color}
\usepackage[all]{xy}

\usepackage{theorem}
\newtheorem{thm}{Theorem}
\newtheorem{lem}[thm]{Lemma}

{\theorembodyfont{\upshape}
\newtheorem{defin}[thm]{Definition}}

{\theorembodyfont{\upshape}
\newtheorem{exmp}[thm]{Example}}
{\theorembodyfont{\upshape}
}
{\theorembodyfont{\upshape}
\newtheorem{protocol}[thm]{Protocol}}

\numberwithin{thm}{section}
\numberwithin{equation}{section}

\newcommand{\Skeptic}{\textsf{Skeptic}}
\newcommand{\Reality}{\textsf{Reality}}

\journalname{Journal of Statistical Physics}

\begin{document}

\title{Gibbs Distribution From Sequentially Predictive Form of the Second Law}

\author{Ken Hiura}

\institute{K. Hiura \at
              Department of Physics, Kyoto University, Kyoto 606-8502, Japan \\
              \email{hiura.ken.88n@st.kyoto-u.ac.jp} }

\date{Received: date / Accepted: date}

\maketitle

\begin{abstract}
We propose a prequential or sequentially predictive formulation of the work extraction where an external agent repeats the extraction of work from a heat engine by cyclic operations based on his predictive strategy. We show that if we impose the second law of thermodynamics in this situation, the empirical distribution of the initial microscopic states of the engine must converge to the Gibbs distribution of the initial Hamiltonian under some strategy, even though no probability distribution are assumed. We also propose a protocol where the agent can change only a small number of control parameters linearly coupled to the conjugate variables. We find that in the restricted situation the prequential form of the second law of thermodynamics implies the strong law of large numbers of the conjugate variables with respect to the control parameters. Finally, we provide a game-theoretic interpretation of our formulation and find that the prequential work extraction can be interpreted as a testing procedure for random number generator of the Gibbs distribution.

\end{abstract}

\section{Introduction}


Equilibrium statistical mechanics gives a tool for calculating macroscopic thermodynamic quantities from the Hamiltonian characterizing microscopic properties of the system \cite{callen}. A fundamental assumption of equilibrium statistical mechanics is that microscopic states are randomly sampled according to the Gibbs distribution for the Hamiltonian. While we can obtain statistical properties of observables, such as means and variances, from this probabilistic assumption, the assumption is also consistent with the second law of thermodynamics. In fact, we cannot extract a strictly positive amount of work through any cyclic process \textit{on average} if the initial probability distribution is Gibbssian \cite{pw,lenard,gp,daniels,jarzynski,crooks}. This result can be regarded as a derivation of the second law of thermodynamics from statistical mechanics. It is natural to ask whether the second law conversely characterizes the Gibbs distribution or not. This question has been traditionally studied in terms of passivity \cite{pw,lenard,gp,daniels}. These studies showed that the initial probability distribution is Gibbssian if and only if any number of copies of the identical state satisfy the second law of thermodynamics. In this approach, we crucially assume that the system is described by a probability distribution on the phase space. In contrast, our question we consider in this paper is how the probabilistic description based on the Gibbs distribution emerges from the second law of thermodynamics, particularly the absence of the perpetual motion machine of the second kind. 

A similar question was posed in the context of probability theory. In measure-theoretic probability theory, we assign a real value in $[0,1]$ to each event under certain compatibility conditions such as the additivity. Although measure-theoretic probability theory is a useful tool to analyze the random behavior of phenomena in nature, it does not provide a characterization of randomness itself. As an alternative approach, Shafer and Vovk proposed game-theoretic probability theory \cite{sv1,sv2}. In a gambling, they say that the gambling is fair if the gambler never become infinitely rich in the limit as the gamble continues and that the sequence of events obtained from such gambling is random. Based on these ideas, they proved, say, the law of large numbers in terms of the gambling, without using measure-theoretic concepts. If we think of the work extraction as a certain gambling between an agent and the nature, the second law of thermodynamics corresponds to the fairness condition in the gambling. If this reasoning is true, we can expect that randomness of the microscopic states in an equilibrium system is characterized through the second law of thermodynamics without using the probabilistic assumption in statistical mechanics. The purpose of the paper is to validate this idea and to answer in the affirmative.

\subsection{Elementary Example: Single-Particle Ideal-Gas Engine}
\label{subsec:idealgas}

We consider a single-particle ideal-gas engine as an elementary example to clarify our problem. A single particle is confined in a box of volume $V = L^3$ and in contact with a heat bath having a temperature $T$. An external agent attempts to extract work from the system. The agent inserts a barrier at the center of the box, $x = L/2$, and moves the barrier quasi-statically to $x = (1-\mu) L$, where $\mu \in (0,1)$ specifies the final position of the barrier. When the particle is on the left side (resp. right side), we set $\omega = 0$ (resp. $\omega = 1$). Since the volume of the region in which the particle is confined after the operation is $\mu^{\omega}(1-\mu)^{1-\omega} V$, the work extracted in this process is given by
\begin{align} \label{eq:spigw}
 W^{\mu}(\omega) = \int_{V/2}^{\mu^{\omega}(1-\mu)^{1-\omega}V} \frac{k_{\mathrm{B}}T}{V} dV = k_{\mathrm{B}}T \ln 2 \mu^{\omega} (1-\mu)^{1-\omega},
\end{align}
where $k_{\mathrm{B}}$ is the Boltzmann constant and we have used the equation of state $P(T,V,N) = N k_{\mathrm{B}}T/V$ for the ideal gas. Finally, the agent removes the barrier and the system returns to the initial state so that the overall process becomes cyclic. According to the second law of thermodynamics, the mechanical work extracted by any cyclic operations is always non-positive. Eq. (\ref{eq:spigw}), however, becomes positive for some $\mu$ and $\omega$, e.g., $\mu = 1/4$ and $\omega = 0$. In statistical mechanics and stochastic thermodynamics, this apparent inconsistency is considered to arise from the fluctuations in small systems. To resolve the inconsistency, we assume that $\omega$ is a random variable obeying the equilibrium distribution $P_{1/2} (\{ \omega = 0 \}) = P_{1/2} (\{ \omega = 1\}) = 1/2$. Then, although the second law is violated with positive probability for $\mu \neq 1/2$, $P_{1/2}(W^{\mu} > 0) = 1/2 > 0$, due to the fluctuation, the \textit{expectation value} of the extracted work $W^{\mu}$ is non-positive for any operations $\mu$, $\mathbb{E}_{P_{1/2}} [W^{\mu}] = k_{\mathrm{B}}T \ln 4 \mu (1-\mu) \leq 0$, and the second law of thermodynamics remains true \textit{on average}. Moreover, the second law conversely characterizes the equilibrium distribution. To prove this statement, we assume that the initial distribution of $\omega$ is given by $P_{\rho} ( \{ \omega = 0 \} ) = 1- \rho$ and $P_{\rho}(\{ \omega = 1 \} ) = \rho$, where $\rho \in (0,1)$ quantifies the inhomogeneity of the particle distribution. The expectation value of the extracted work is given by $\mathbb{E}_{P_{\rho}} [W^{\mu}] = k_{\mathrm{B}} T ( D(P_{\rho} \| P_{1/2}) - D(P_{\rho} \| P_{\mu}))$. Here $D(P_{\rho} \| P_{\mu}) \coloneqq \rho \ln (\rho / \mu) + (1-\rho) \ln [(1-\rho) / (1-\mu)]$ is the Kullback-Leibler divergence between the Bernoulli distributions $P_{\rho}$ and $P_{\mu}$. If the initial distribution is not uniform, i.e., $\rho \neq 1/2$, we have that $\mathbb{E}_{P_{\rho}} [W^{\rho}] = k_{\mathrm{B}} T D(P_{\rho} \| P_{1/2}) > 0$ by choosing $\mu = \rho$. Therefore, the condition that $\mathbb{E}_{P_{\mu}} [W^{\mu}] \leq 0$ for any $\mu \in (0,1)$ implies $\rho = 1/2$. In summary, the non-positivity of the averaged work is equivalent to the equilibrium condition for the initial distribution in this example.

We stress again that the extracted work $W^{\mu}(\omega)$ is positive for some initial state $\omega$ and operation $\mu$. Therefore, when we speak of the validity of the second law of thermodynamics in small systems, we have to consider a situation that we prepare microscopic states and extract work many times. Moreover, if the empirical frequency of the initial microscopic states is biased from the equilibrium distribution, the second law may be violated suggested by the argument in the previous paragraph. In this sense, the second law in small systems may requires the stochastic behavior of the initial microscopic states as well as the equilibrium condition for the probability distribution characterizing its stochasticity. In this paper, we ask how can we formulate mathematically the emergence of stochasticity characterized by the equilibrium distribution from the second law of thermodynamics. To answer the question, we investigate two issues in this subsection.

First, we provide a mathematical definition of ``stochastic behavior of the initial microscopic states''. In this paper, we consider the situation that the agent repeats the cyclic operations infinitely many times and ask whether an infinite sequence $\omega_1 \omega_2 \dots \in \{ 0, 1 \}^{\mathbb{N}_{+}}$ of initial microscopic states in the experiments is a random sequence or not with respect to the equilibrium distribution $P_{1/2}$. Here $\mathbb{N}_{+} = \{ 1, 2, \dots \}$ denotes the set of positive integers and $\omega_n \in \{ 0, 1\}$ denotes the initial position of the particle in the $n$-th cycle. An example of non-random sequence is $00000 \dots$, which corresponds to the situation that the particle is always in the left side of the box. Although the theory of algorithmic randomness \cite{nies,dh} provides a reasonable and rigorous definition of randomness for individual sequences, we pay attention to only the convergence of the empirical distribution and empirical mean in this paper. For the single-particle ideal-gas engine, we regard an infinite sequence $\omega_1 \omega_2 \dots$ representing the positions of the particle as a random sequence if the sequence satisfies the strong law of large numbers (SLLN),
\begin{align} \label{eq:slln}
 \lim_{n \to \infty} \frac{1}{n} \sum_{i=1}^{n} \omega_i = \frac{1}{2}.
\end{align}
We discuss the relation to the algorithmic randomness in Section \ref{sec:cr}.

Second, we formulate the second law of thermodynamics without probability measures. Since our purpose is to clarify the emergence of the probabilistic description, we cannot suppose the underlying probability distribution at the starting point of the discussion and use the standard definition of the second law $\mathbb{E} [W^{\mu}] \leq 0$. To remove the probability distribution and expectation value from the description, we consider again the situation that the agent repeats the cyclic operations infinitely many times (Fig. \ref{fig:protocol}). In contrast to averaging the extracted work for a single cyclic operation, in repeated cyclic operations, we allow the agent to choose a different operation in each cycle. A critical problem here is to specify the information that the agent can use when he decides the cyclic operation in each cycle. In this paper, we apply the \textit{prequential scheme} \cite{dawid} or predictive scheme, where the agent determines the volume fraction $\mu_n \in (0,1)$ in the $n$-th cycle depending on only the past history $\omega_1, \dots, \omega_{n-1}$ of the initial positions up to the $(n-1)$-th cycle \footnote{The protocol should not be confused with that in Szilard's engine \cite{szilard}. Imagine that positions of the particle are prepared independently and identically according to the uniform distribution $P_{1/2}$. The mutual information between the positions up to the $(n-1)$-th cycle and the position in the $n$-th cycle is zero due to the statistical independency. Therefore, even if the agent can use the information on the past history, the expectation value of the extracted work in the $n$-th cycle is always non-positive.}. In other words, the agent predicts the position of the particle in the $n$-th cycle from the past results $\omega_1 \dots \omega_{n-1}$ and performs a cyclic operation based on the prediction. An assignment $\omega_1 \dots \omega_{n-1} \mapsto \mu_n$ for each $n = 1, 2, \dots$ represents a prediction scheme called a \textit{strategy} for the agent and denoted by $\hat{\mu}$. For a given strategy $\hat{\mu}$ for the agent, the accumulation of the extracted work $W^{\hat{\mu}}$ is given by
\begin{align}
 W^{\hat{\mu}}(\omega_1 \dots \omega_n) \coloneqq \sum_{i=1}^{n} k_{\mathrm{B}}T \ln 2 \mu_t^{\omega_i} (1-\mu_t)^{1-\omega_i}.
\end{align}
We define the violation of the second law in terms of the asymptotic behavior of $W^{\hat{\mu}}$. We say that an infinite sequence $\omega_1 \omega_2 \dots \in \{ 0, 1\}^{\mathbb{N}_+}$ of initial positions of the particle violates the second law of thermodynamics under the strategy $\hat{\mu}$ if the total amount of the extracted work from the heat engine diverges to infinity, i.e.,
\begin{align} \label{eq:psl}
 \lim_{n \to \infty} W^{\hat{\mu}} (\omega_1 \dots \omega_n) = \infty.
\end{align}
This means that the agent can extract work from such a sequence as much as he wants by repeating the cyclic operations sufficiently many times. For example, if the agent chooses $\mu_t = 1/4$ for any $t \in \mathbb{N}_+$, the infinite sequence $0000 \dots$ violates the second law because $W^{\hat{\mu}}(\omega_1 \dots \omega_n) = n k_{\mathrm{B}}T \ln (3/2)$ diverges to infinity as $n \to \infty$. We adopt this definition as the second law in our study because it refers to no probability measure. We note that this definition is consistent with equilibrium statistical mechanics. In fact, if we assume that $\omega_1, \omega_2, \dots$ are independent and identically distributed random variables obeying the product distribution $P_{1/2}^{\otimes \mathbb{N}_{+}}(\{ \omega_n = 0 \}) = P_{1/2}^{\otimes \mathbb{N}_{+}} (\{ \omega_n = 1 \})=1/2$, the probability that the second law is violated is zero,
\begin{align} \label{eq:pv2l}
 P_{1/2}^{\otimes \mathbb{N}_{+}} \left\{ \lim_{n \to \infty} W^{\hat{\mu}}_n =  \infty \right\} = 0,
\end{align}
where we have defined a random variable $W_n^{\hat{\mu}} (\xi) \coloneqq W^{\hat{\mu}}(\omega_1 \dots \omega_n)$ for $\xi = \omega_1 \omega_2 \dots \in \{ 0, 1 \}^{\mathbb{N}_{+}}$. We present a proof of (\ref{eq:pv2l}) in Appendix \ref{ap:pv2l}.

\begin{figure}
\centering
\includegraphics[width=11cm]{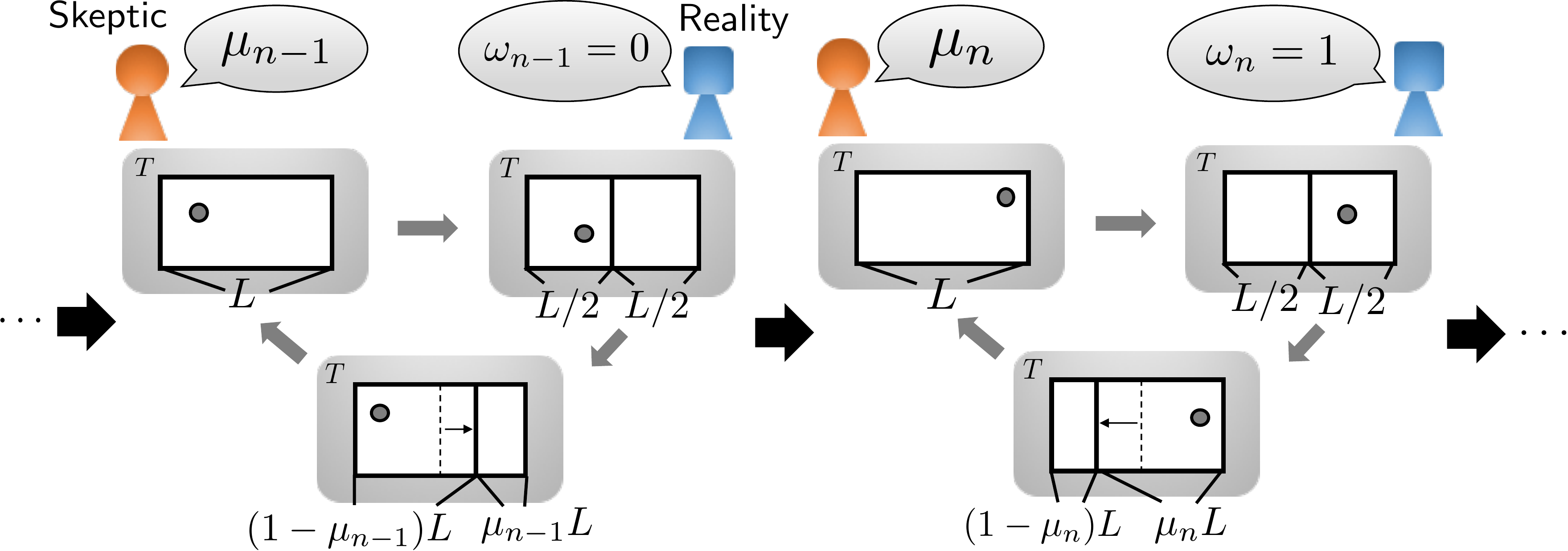}
\caption{Schematic of the protocol. (1) An agent named $\Skeptic$ announces the position of the barrier $\mu_n \in (0,1)$ depending on the past history $\omega_1 \dots \omega_{n-1}$ of positions of the particle. (2) Another agent named $\Reality$ announces the position of the particle $\omega_n$ in the $n$-th cycle. (3) $\Skeptic$ moves the barrier from $x = L/2$ to $x = (1-\mu_n)L$ and remove it. (4) Go back to (1). }
\label{fig:protocol}
\end{figure}

Based on the above two arguments, our question is formulated as follows. Instead of assuming that $\omega_n$ is a random variable and introducing a probability measure from the beginning, we ask which sequences can be realized under the second law of thermodynamics and what statistical properties such sequences have commonly. As we noted before, in this paper, we consider the strong law of large numbers (\ref{eq:slln}) as a relevant statistical property. The SLLN is formulated in probability theory as
\begin{align}
 P_{1/2}^{\otimes \mathbb{N}_{+}} \left\{ \lim_{n \to \infty} \frac{1}{n} \sum_{i=1}^{n} \omega_i = \frac{1}{2} \right\} = 1.
\end{align}
This statement means that the relative frequencies of positions of particle becomes one half almost surely. However, we give a different formulation of the SLLN in our setting because no probability measure enters the setting. According to the excellent book ``Probability and Finance: It's Only a Game!'' by Shafer and Vovk \cite{sv1}, there exists a strategy $\hat{\mu}_{\mathrm{SLLN}}$ for the agent such that if an infinite sequence $\omega_1 \omega_2 \dots \in \{ 0, 1\}^{\mathbb{N}_+}$ of positions of the particle does \textit{not} satisfy the SLLN (\ref{eq:slln}), the second law of thermodynamics is then violated, i.e., (\ref{eq:psl}) holds for the sequence. Equivalently, if the sequence retains the second law under the strategy $\hat{\mu}_{\mathrm{SLLN}}$, the positions of the particle necessarily obey the SLLN (\ref{eq:slln}). This implies that even though no probability measure is assumed, the empirical distribution for positions of the particle must be consistent with the equilibrium distribution due to the constraint by the second law. From the viewpoint of statistical mechanics, the second law of thermodynamics is a consequence of the equilibrium distribution and equilibrium statistical mechanics gives a microscopic foundation for the second law. However, according to the theorem by Shafer and Vovk, the second law of thermodynamics (\ref{eq:psl}) under some strategy requires that the sequence must be random in the sense that it satisfies the SLLN (\ref{eq:slln}), and leads to the equilibrium distribution in this sense. Our purpose of the paper is to investigate this novel aspect of the relationship between equilibrium statistical mechanics and the second law of thermodynamics for generic small systems.

\subsection{Main Results}

The argument in subsection \ref{subsec:idealgas} suggests that statistical properties of equilibrium states are characterized by the second law of thermodynamics. In this paper, we extend the above example to a system with a generic Hamiltonian on a finite state space. First, we consider the situation that the external agent has an ability to prepare arbitrary Hamiltonians during cyclic operations. We show that there exists a strategy for the agent such that the empirical distribution for a sequence satisfying the second law of thermodynamics necessarily converges to the Gibbs distribution of the initial Hamiltonian (Theorem \ref{thm:np}). As in the case of the single-particle ideal-gas engine, this result can be interpreted as a statement that the empirical statistics must be consistent with the assumption of statistical mechanics, i.e., the Gibbs distribution, due to the second law.


Second, we study the empirical statistics in the same manner when we restrict the ability of the agent. The assumption that the agent can prepare arbitrary Hamiltonians is too demanding for his ability because we control only a small number of parameters in the Hamiltonian, such as the magnetic field, in many physical situations.  The main aim of the present paper is to propose a protocol corresponding to such a  restricted situation and to determine what statistical properties are observed in that situation. This restriction weakens the ability of the agent, and we thus expect that the statistical property the second law imposes also weakens. We find that the Gibbs distributions for Hamiltonians having a small number of parameters linearly coupled to conjugate variables form an exponential family and propose a new protocol where the agent has to construct his strategy by changing only these parameters. Our main contribution is that in the protocol there exists a strategy such that the empirical mean of the conjugate variable with respect to the control parameter for sequence satisfying the second law converges to the equilibrium value (Theorem \ref{thm:p}). This result suggests that there is a hierarchy of statistical properties observed under the second law according to his ability.

\subsection{Related Studies}

Several studies shares the same mathematical structures and techniques with our work, although our motivation and formalism in this paper is quite different from the usual one in statistical mechanics and stochastic thermodynamics. We review related studies in this subsection.

First, an important property of the exponential of the accumulation of the extracted work is martingality, which is a fundamental concept in the theory of stochastic processes \cite{williams,doob}. The martingale property is useful to investigate the statistics at stopping times and extreme value statistics. In the context of nonequilibrium thermodynamics, the novel statistical properties of stochastic entropy production were recently studied \cite{cg,nrj,neri,msmfpr} based on this property of martingales and the fact that the exponentiated negative entropy production or its modification is a martingale. Although our results are also based on the fact that the exponential of the accumulation of the extracted work is a martingale\footnote{In externally driven systems, the exponentiated negative entropy production in a time interval is not martingale in general \cite{cg,neri,msmfpr}. In this paper, however, we do not consider the stochastic time evolution of the system explicitly and concentrate our interest on the sum of the extracted work obtained from statistically independent experiments. Therefore, the exponentiated negative entropy production in this paper is indeed a martingale. See Appendix \ref{ap:mg}.}, we use another property that was first found by Jean Ville \cite{ville}, the characterization of almost sure properties in terms of martingales. See subsection \ref{subsec:ville}. As other interesting study concerning the martingale property in physical systems, see \cite{ms} for example.

Second, we use universal coding theory \cite{grunwald} to construct a strategy for the agent. The prediction strategy in universal coding is useful for proving our main result because our problem is similar to the coding or prediction of the outcomes $\omega_n$ in the $n$-th cycle from the past sample data $\omega_1 \dots \omega_{n-1}$ where the performance of the prediction is measured by the log-loss function. The analogy between gambling and source coding problems was first discussed by Kelly \cite{kelly}. See also Chapter 6 of \cite{ct}. In addition, the analogy with information thermodynamics was pointed out in \cite{vpm,ito,tmmr}. Specifically, Refs. \cite{tmmr} applies universal coding theory to information thermodynamics to construct an optimal work extraction protocol. A crucial difference of our work from this study is that we analyze the asymptotic behavior of extracted work for \textit{individual} sequences \cite{vpm,kg} and attempt to find statistical properties shared by sequences satisfying the second law.

Finally, we stress that our studies are based on an analogy between the work extraction in thermodynamics and betting in game-theoretic probability theory \cite{sv1,sv2}. Game-theoretic probability theory is a mathematical formulation of probability theory alternative to the conventional measure-theoretic one. In game-theoretic probability theory, a gambler named $\Skeptic$ bets money on head or tail of a coin and a dealer named $\Reality$ choose the outcome. By repeating this gamble infinitely many times and imposing the duty that $\Reality$ must chooses a sequence of outcomes such that $\Skeptic$ cannot make infinitely much money, we study what statistical behavior is observed. For instance, as mentioned in subsection \ref{subsec:idealgas}, it is possible to construct an explicit strategy such that $\Skeptic$'s capital grows infinitely as long as the sequence violates the law of large numbers. The work extraction we proposed in this paper can be regarded as a game played between two players, $\Skeptic$ and $\Reality$. The external agent who attempts to extract work from the heat engine as much as possible corresponds to the gambler $\Skeptic$ and the world who prepares the particle to retain the second law of thermodynamics corresponds to the dealer $\Reality$. Although the first main result of this paper in Section \ref{sec:gh} is a straightforward extension of the theorem proved by Shafer and Vovk, the possibility of analogous analysis of thermodynamics and the novel aspect of the relationship between statistical mechanics and the second law are new findings of this paper as long as the author knows.

\subsection{Outline of the paper}

The remainder of the paper is organized as follows. In Section \ref{sec:gh}, we formulate the work extraction in a similar manner to the single-particle ideal-gas engine in subsection \ref{subsec:idealgas} and give the first main result. We also discuss the mathematical backgrounds, Ville's theorem, behind our result. In Section \ref{sec:ph}, we propose another protocol where the ability of the agent is restricted, and prove the second main result. In Section \ref{sec:gi}, we discuss a game-theoretic interpretation of our protocols. We end our paper with concluding remarks in Section \ref{sec:cr}.

\subsection{Notations}

This subsection summarize notations we use throughout this paper. Since the author explain notations when they are first used, the readers can skip this subsection.

Let $\mathbb{N}_+ = \{ 1, 2, \dots \}$ be the set of positive integers. We use $\Omega$ to denote a finite set representing a microscopic state space. $\Omega^* \coloneqq \{ \square \} \cup (\cup_{n=1}^{\infty} \Omega^n)$ denotes the set of finite strings over $\Omega$, where $\square$ is the empty string. A string over $\Omega$ with length $n$ is written as $\omega^n = \omega_1 \dots \omega_n \in \Omega^*$, $\omega_i \in \Omega$. We use $\xi = \omega_1 \omega_2 \dots \in \Omega^{\mathbb{N}_{+}}$ to denote an infinite sequence on $\Omega$. A real-valued function $H:\Omega \to \mathbb{R}$ defines a Hamiltonian on the state space $\Omega$. For a positive real number $\beta > 0$, the Gibbs distribution for the Hamiltonian $H$ at the inverse temperature $\beta$ is defined as a probability distribution on $\Omega$ with density $g_{\beta H} (\omega) \coloneqq e^{- \beta (H(\omega) - F_{\beta}(H))}$ with respect to the counting measure, where $F_{\beta}(H) \coloneqq - \beta^{-1} \ln \sum_{\omega \in \Omega} e^{- \beta H(\omega)}$ is the free energy.

\section{Analysis of Generic Hamiltonians}
\label{sec:gh}

\subsection{Setup}
\label{subsec:setup}

Let us consider a physical system whose state space is given by a finite set $\Omega$. The thermodynamic property of the  system in contact with the heat bath is described by a Hamiltonian $H : \Omega \to \mathbb{R}$ and an inverse temperature $\beta$ of the bath. According to Kelvin's principle, which is one of the representation of the second law of thermodynamics, the positive amount of work cannot be extracted by any cyclic operations. Here an operation is said to be cyclic if the initial and final Hamiltonians coincide. This principle leads to the absence of the perpetual motion machine of the second kind. However, the second law of thermodynamics may be violated for some individual initial state and cyclic operation in the finite system. Therefore, the second law in small systems is usually formulated as a statement on the non-positivity of the \textit{expectation value} of the extracted work assuming the underlying probability distribution. In this subsection, we review the usual formulation of the second law of thermodynamics in terms of the expectation value.

Let us consider the following type of cyclic process \cite{evdb} to avoid taking the dynamical evolution of the system into consideration:
\begin{itemize}
 \setlength{\leftskip}{3.0mm}
 \item[(P1)] The agent quenches the Hamiltonian adiabatically from the initial Hamiltonian $H$ to another one $H^{\prime}$.
 \item[(P2)] The agent equilibrates the system with the inverse temperature $\beta$.
 \item[(P3)] The agent resets the system quasi-statically and isothermally
\end{itemize}
Let $\omega$ be an initial state of the system. We suppose that the microscopic state $\omega$ does not change during the adiabatic quenching process (P1). Under this assumption, the extracted work in the process (P1) is given by the decrease in internal energy $H(\omega) - H^{\prime}(\omega)$. In the process (P2), the agent touches the system with the heat bath having the inverse temperature $\beta$ and the system relaxes to the new equilibrium state for the quenched Hamiltonian $H^{\prime}$. This equilibration process obviously requires no mechanical work. In the process (P3), the agent changes the Hamiltonian from the quenched one $H^{\prime}$ to the initial one $H$ quasi-statically to make the whole process cyclic. Moreover, we crucially assume that the work extracted in the quasi-static isothermal process is equal to the decrease in free energy\footnote{We discuss the status of this assumption after Theorem \ref{thm:np} and Section \ref{sec:cr}.}. This assumption implies that the extracted work in the process (P3) is given by $F_{\beta}(H^{\prime}) - F_{\beta}(H)$. Here the free energy for the Hamiltonian $H$ at the inverse temperature $\beta$ is defined as $F_{\beta}(H) \coloneqq - \beta^{-1} \ln \sum_{\omega \in \Omega} e^{- \beta H(\omega)}$. Therefore, the total amount of work extracted in this cyclic process is given by $W(\omega_1) \coloneqq H(\omega) - H^{\prime}(\omega) + F_{\beta}(H^{\prime}) - F_{\beta}(H)$.

Now we suppose that the initial state is sampled according to an initial density $\rho$. The expectation value of the work extracted during the above cyclic process is given by
\begin{align}
 \mathbb{E}_{\rho} [ W ] = \sum_{\omega \in \Omega} \rho (\omega) \left[ H(\omega) - H^{\prime}(\omega) + F_{\beta}(H^{\prime}) - F_{\beta}(H) \right] = D ( \rho \| g_{\beta H} ) - D(\rho \| g_{\beta H^{\prime}}).
\end{align}
Here $g_{\beta H} (\omega) \coloneqq e^{- \beta (H(\omega) - F_{\beta}(H))}$ is the Gibbs density function for the Hamiltonian with the inverse temperature $\beta$ and $D(p \| q) \coloneqq \sum_{\omega \in \Omega} p(\omega) \ln (p(\omega) / q(\omega))$ is the Kullback-Leibler divergence between two densities $p$ and $q$. If the initial distribution is the Gibbs distribution for the initial Hamiltonian, the expectation value of the extracted work is non-positive for any quenched Hamiltonian $H^{\prime}$, i.e., $\mathbb{E}_{\rho} [ W ] = - D(\rho \| g_{\beta H^{\prime}} ) \leq 0$. Conversely, if $\rho \neq g_{\beta H}$, the value $\mathbb{E}_{\rho}[W]$ can be positive. Indeed, by choosing $H^{\prime}$ such that $\rho = g_{\beta H^{\prime}}$, we have that $\mathbb{E}_{\rho}[W] = D(\rho \| g_{\beta H}) > 0$. Therefore, the second law of thermodynamics expressed in the form of the expectation value during the cyclic process (P1)-(P3) is equivalent to that the initial distribution is the Gibbs distribution for the initial Hamiltonian.

The single-particle ideal-gas engine in subsection \ref{subsec:idealgas} is formally considered to be an example of the above setup. 

\begin{exmp}[Single-particle ideal-gas engine]
Let $\Omega = \{ 0, 1 \}$ be a state space. Each state $\omega \in \{ 0, 1 \}$ codes the position of the particle in the box. The effective Hamiltonian for the single-particle ideal gas is defined as $H_{\mu} ( \omega ) = - \beta^{-1} \ln \mu^{\omega} (1-\mu)^{1-\omega} V$, where $V$ is the total volume of the box and the parameter $\mu \in (0,1)$ indicates the position of the barrier. The free energy for the Hamiltonian is given by $F_{\beta}(H_{\mu}) = - \beta^{-1} \ln V$. If the initial and quenched parameters are $1/2$ and $\mu$, respectively, the extracted work is written as
\begin{align}
 W(\omega) = H_{1/2}(\omega) - H_{\mu}(\omega) + F_{\beta}(H_{\mu}) - F_{\beta}(H_{1/2}) = \beta^{-1} \ln 2 \mu^{\omega} (1-\mu)^{1-\omega},
\end{align}
which is identical to Eq. (\ref{eq:spigw}).
\end{exmp}

\subsection{Prequential formulation}
\label{subsec:preq}

The purpose of the present paper is to investigate the emergence of equilibrium statistical mechanics from the second law of thermodynamics without referring to any probability measure. Hence we have to remove the probability measure from the definition of the second law. In this subsection, as such a formulation, we give a prequential definition of the second law.

To remove the probability distribution, we consider the situation that the agent repeats cyclic processes (P1)-(P3) infinitely many times. First, the agent performs a cyclic process according to the protocol (P1)-(P3). Let $\omega_1$ be an initial state and $H_1$ a quenched Hamiltonian in this process. After the first cycle, the agent determines a quenched Hamiltonian $H_2$ in the second cycle depending on the initial state $\omega_1$ in the first cycle and performs the cyclic process (P1)-(P3) again. In general, we suppose that the agent chooses a quenched Hamiltonian in the $n$-th cycle depending on the past history $\omega_1 \dots \omega_{n-1} \in \Omega^{n-1}$ up to the $(n-1)$-th cycle. The assignment of a quenched Hamiltonian in the $n$-th cycle to each past history $\omega_1 \dots \omega_{n-1}$ specifies a strategy for the agent to extract work. Therefore, we call a function $\hat{H} : \Omega^* \to \mathbb{R}^{\Omega}$ \textit{strategy} in this paper. Here $\Omega^*$ denotes the set of finite strings over $\Omega$ including the empty string $\square$. For a strategy $\hat{H}$, $\hat{H}(\cdot | \square): \Omega \to \mathbb{R}$ represents a quenched Hamiltonian in the first cycle and $\hat{H}(\cdot | \omega_1 \dots \omega_{n-1}): \Omega \to \mathbb{R}$ represents a quenched Hamiltonian in the $n$-th cycle when the initial states up to the $(n-1)$-th cycle are $\omega_1 \dots \omega_{n-1}$. This scheme in which the agent decides his action based on the past history of outcomes is called \textit{prequential} in statistics \cite{dawid} and \textit{causal} or \textit{nonanticipating} in information theory of gambling and portfolio theory \cite{ct}.

The accumulation of the extracted work up to the $n$-th cycle is given by the sum of the extracted work in each cycle. For a strategy $\hat{H}$, we define the function $W^{\hat{H}} : \Omega^* \to \mathbb{R}$ as $W^{\hat{H}}(\square) = 0$ and
\begin{align} \label{eq:work}
 W^{\hat{H}}(\omega^n) &\coloneqq \sum_{i=1}^{n} \left[ H(\omega_i) - \hat{H} (\omega_i | \omega^{i-1}) + F_{\beta}(\hat{H}(\cdot | \omega^{i-1})) - F_{\beta}(H) \right]
\end{align}
for $\omega^n \coloneqq \omega_1 \dots \omega_n \in \Omega^n$. From the definition, $W^{\hat{H}}(\omega^n)$ gives the accumulation of the work up to the $n$-th cycle under the strategy $\hat{H}$ when the initial states up to the $n$-th cycle are $\omega_1 \dots \omega_n$.

Now we provide a prequential definition of the second law of thermodynamics. 

\begin{defin} \label{def:2l}
Let $\xi = \omega_1 \omega_2 \dots \in \Omega^{\mathbb{N}_+}$ be an infinite sequence over $\Omega$ and $\hat{H}$ be a strategy for the agent. We say that
\begin{itemize}
 \item[(1)] $\xi$ \textit{violates weakly} the second law of thermodynamics under the strategy $\hat{H}$ if
 \begin{align}
 \sup_{n} W^{\hat{H}} (\omega_1 \dots \omega_n) = \infty,
 \end{align}
 \item[(2)] $\xi$ \textit{violates} the second law of thermodynamics under the strategy $\hat{H}$ if
 \begin{align}
 \lim_{n \to \infty} W^{\hat{H}} (\omega_1 \dots \omega_n) = \infty.
 \end{align}
\end{itemize}
\end{defin}

Definition \ref{def:2l} is regarded as a definition of the perpetual motion machine of the second kind for individual sequences of states. Let $\omega_1$ be an initial microscopic state of the engine in the first cycle. By performing the cyclic process (P1)-(P3) for the initial state, the agent extract work from the engine by $W^{\hat{H}}(\omega_1)$. The engine cannot be regarded as a perpetual motion machine of the second kind only because the value $W^{\hat{H}}(\omega_1)$ is positive. To say that the engine violates the second law, we require that for any given positive value $W_0 > 0$, the agent should extract an amount of work larger than $W_0$ by repeating cyclic processes as many times as he needs. Therefore, we define the second law for \textit{infinite sequences of initial microscopic states of the engine} as indicated in Definition \ref{def:2l}. We note that while in the above definition the violation of the second law depends on both a strategy the agent applies and an infinite sequences of microscopic states, the definition needs no underlying probability measure.

\subsection{Convergence of empirical distribution to Gibbs distribution}

Instead of introducing probability distributions, we consider which sequences satisfies the second law of thermodynamics and what statistical properties are shared among these sequences. In general, a statistical property is described by a subset of infinite sequences $E \subseteq \Omega^{\mathbb{N}_{+}}$. In this paper, as a relevant statistical property, we focus on only the strong law of large numbers, i.e.,
\begin{align} \label{eq:gslln}
 \mathsf{SLLN} = \left\{ \xi = \omega_1 \omega_2 \dots \in \Omega^{\mathbb{N}_{+}} : \lim_{n \to \infty} \frac{1}{n} \sum_{i=1}^{n} 1_{\{ \omega \}} (\omega_i) = g_{\beta H} (\omega) \ \text{for all } \omega \in \Omega \right\}.
\end{align}
Here the quantity $\sum_{i=1}^{n} 1_{\{ \omega \}}(\omega_i) / n$ is the empirical density for the string $\omega^n$ quantifying the relative frequency of $\omega$ in $\omega^n = \omega_1 \dots \omega_n$. If the sequence $\xi = \omega_1 \omega_2 \dots$ are sampled identically and independently according to the Gibbs distribution, the SLLN (\ref{eq:gslln}) happens almost surely. In this sense, the SLLN (\ref{eq:gslln}) is a statistical property observed in equilibrium. Our purpose is to find a strategy $\hat{H}$ such that any sequences $\xi \in \Omega^{\mathbb{N}_{+}}$ satisfying the second law of thermodynamics under the strategy $\hat{H}$ in the sense of Definition \ref{def:2l} necessarily obey the SLLN, i.e., $\xi \in \mathsf{SLLN}$. The existence of such strategy implies that the empirical statistics of the sequences of initial states observed under the constraint by the second law must be consistent with equilibrium statistical mechanics. Based on the above reasoning, we provide the following theorem:

\begin{thm} \label{thm:np}
Let $g_{\beta H} (\omega) = e^{- \beta H(\omega) + \beta F_{\beta}(H)}$ be the Gibbs density function for the initial Hamiltonian $H$ at the inverse temperature $\beta$. There exists a strategy $\hat{H}$ such that if an infinite sequences $\xi \in \Omega^{\mathbb{N}_{+}}$ of initial states does not satisfy the SLLN, i.e., $\xi \not\in \mathsf{SLLN}$, then the sequence $\xi$ violates the second law of thermodynamics under the strategy.
\end{thm}

This is the first main result of this paper. We make several remarks on Theorem \ref{thm:np}. First, we stress again that no probability measure for initial microscopic states are assumed. One may find that the Gibbs distribution is implicitly inserted in the definition of the extracted work (\ref{eq:work}) through the assumption that the work extracted in the quasi-static isothermal process is equal to the decrease in free energy. Although this reasoning is actually true as we will see in subsection \ref{subsec:ville}, the assumption on the form of the extracted work does \textit{not} immediately leads to the stochastic behavior of initial states and the content of Theorem \ref{thm:np} remains highly non-trivial. Second, we stress the difference from the argument in the end of subsection \ref{subsec:setup}. There, if the initial distribution deviates from the Gibbs distribution for the initial Hamiltonian, we have to choose the quenched Hamiltonian depending on the deviation in order to extract a positive amount of work. In contrast, Theorem \ref{thm:np} claims the existence of a single \textit{universal} strategy under which the second law is automatically violated for the sequence whose empirical statistics deviate from the Gibbs distribution even if we do not know the deviation.

We explain in subsection \ref{subsec:ville} the reason why the Gibbs distribution appears in Theorem \ref{thm:np} although no probability measure is assumed in our setting. We remark that Theorem \ref{thm:np} is just a straightforward extension of the case of single-particle ideal-gas engine in subsection \ref{subsec:idealgas} and it is nothing new mathematically. In addition to the proof in textbooks of Shafer and Vovk \cite{sv1,sv2}, there are several proofs of Theorem \ref{thm:np} such as Ref. \cite{kt} based on the maximum likelihood strategy and Ref. \cite{ktt} based on the Bayesian strategy. Nevertheless we prove Theorem \ref{thm:np} as a special case of Theorem \ref{thm:p}.

Hereafter, we say that a strategy \textit{forces} (resp. \textit{weakly forces}) an event $E \subseteq \Omega^{\mathbb{N}_{+}}$ if infinite sequences over $\Omega$ that do not satisfy $E$ violate (resp. weakly violate) the second of thermodynamics under strategy. According to this terminology, Theorem \ref{thm:np} claims the existence of a strategy that forces the strong law of large numbers for the empirical distribution.

\subsection{Martingale and Ville's theorem}
\label{subsec:ville}

We clarify a general structure behind Theorem \ref{thm:np}. Although we construct explicitly a strategy that make the empirical distribution converge the Gibbs distribution in Appendix \ref{ap:proof}, Theorem \ref{thm:np} follows from a more general theorem proved by Ville \cite{ville}. The essence of Ville's theorem is some kind of equivalence between the asymptotic behavior of martingales and almost sure properties.

To see this, we first clarify the martingale property of the exponentiated work. The accumulation of the extracted work (\ref{eq:work}) can be written as the logarithmic likelihood ratio function for Gibbs distributions:
\begin{align} \label{eq:llr}
 \beta W^{\hat{H}} (\omega^n) = \ln \frac{\hat{q} (\omega^n)}{g_{\beta H}^n (\omega^n)},
\end{align} 
where $g_{\beta H}^n (\omega^n) \coloneqq \prod_{i=1}^{n} g_{\beta H} (\omega_i)$ and 
\begin{align} \label{eq:qhat}
 \hat{q} (\omega^n) \coloneqq \prod_{i=1}^n g_{\beta \hat{H}(\cdot | \omega^{i-1})}(\omega_i).
\end{align}
Since $g_{\beta \hat{H}(\cdot | \omega^{i-1})}$ specifies the conditional probability density conditioned on the past history $\omega^{i-1}$, the function (\ref{eq:qhat}) gives a probability density on $\Omega^n$. Conversely, a stochastic process on $\Omega^{\mathbb{N}_{+}}$ with strictly positive probability densities $\hat{q}$ specifies a strategy for the agent in our setting through the relation $g_{\beta \hat{H}(\cdot | \omega^{n-1})} (\omega_n) = \hat{q} ( \omega_n | \omega^{n-1})$. Therefore, a strategy $\hat{H}$ is identified with a stochastic process having strictly positive densities.

Let us consider a discrete time stochastic process $M: \Omega^* \to \mathbb{R}$. We say that $M$ is a martingale with respect to a probability measure $P$ on $\Omega^{\mathbb{N}_{+}}$ if $\mathbb{E}_{P} [ M_n | \omega^{n-1} ] = M(\omega^{n-1})$ for any $\omega^{n-1} \in \Omega^*$ \cite{williams}. Here $M_n (\omega_1 \omega_2 \dots) \coloneqq M(\omega^n)$ and $\mathbb{E}_P [ \ \cdot \ | \omega^n]$ denotes the conditional expectation conditioned on the past history $\omega^n$. For a fixed strategy $\hat{H}$, it is easy to see that the exponential of the accumulation of the extracted work $e^{\beta W^{\hat{H}}}$ is a positive martingale with respect to the infinite product of Gibbs distributions $g_{\beta H}^{\mathbb{N}_+}$,
\begin{align*}
 \mathbb{E}_{g_{\beta H}^{\mathbb{N}_+}} [ e^{\beta W_n^{\hat{H}}} \mid \omega^{n-1} ] = e^{\beta W^{\hat{H}}(\omega^{n-1})},
\end{align*}
In particular, $\mathbb{E}_{g_{\beta H}^{\mathbb{N}_+}} [ e^{\beta W^{\hat{H}}_n} ] = e^{\beta W^{\hat{H}}(\square)} = 1$. Conversely, for a given positive martingale $M$ with respect to $g_{\beta H}^{\mathbb{N}_{+}}$ starting from $M(\square) = 1$, there exists a strategy $\hat{H}$ such that $e^{\beta W^{\hat{H}}(\omega^n)} = M(\omega^n)$ for any $\omega^n \in \Omega^*$. See Appendix \ref{ap:mg} for a proof of this statement.

Ville's theorem \cite{ville,sv1,sv2} in our setting claims that \textit{a measurable set $E \subseteq \Omega^{\mathbb{N}_{+}}$ has probability one with respect to $g_{\beta H}^{\mathbb{N}_{+}}$ if and only if there exists a positive martingale $M: \Omega^* \to \mathbb{R}_{+}$ with respect to the same probability measure such that $\lim_{n \to \infty} M_n (\xi) = \infty$ for any $\xi  \not\in E$}. Using the equivalence between positive martingales and strategies mentioned above, we rephrase Ville's theorem as follows.

\begin{thm}[\cite{ville,sv1,sv2}] \label{thm:ville}
Let $E \subseteq \Omega^{\mathbb{N}_{+}}$ be a measurable set. The property $E$ happens almost surely with respect to $g_{\beta H}^{\mathbb{N}_{+}}$ if and only if there exists a strategy $\hat{H} : \Omega^* \to \mathbb{R}^{\Omega}$ that forces $E$.
\end{thm}

Ville's theorem clarifies the reason why the convergence to the Gibbs distribution occurs under the second law in spite that no probability measure comes into our setting. Let us consider the experiment that we sample initial microscopic states infinitely many times independently and identically according to the Gibbs distribution $g_{\beta H}$. We take a statistical property $E \subseteq \Omega^{\mathbb{N}_{+}}$ with probability one in the experiment such as the SLLN (\ref{eq:gslln}). According to Ville's theorem, there exists a strategy $\hat{H}$ such that any sequences violate the second law under the strategy if they do not have the property $E$. In other words, whether a statistical property is consistent with equilibrium statistical mechanics or not is characterized in terms of the violation of the second law of thermodynamics. In addition, our definition of the second law (Definition \ref{def:2l}) needs no probability measure. This is an essential structure that makes it possible to argue the emergence of equilibrium statistical mechanics from the second law without referring to any probability measures.

\section{Analysis of Parametric Hamiltonian}
\label{sec:ph}

In the setting in Section \ref{sec:gh}, the agent is assumed to have an ability to prepare any complicated Hamiltonian containing non-local and many-body interactions. As a more physical situation, it is natural to restrict the possible operations into changing a small number of parameters in a certain class of Hamiltonians. In this section, we investigate such situation.

\subsection{Preliminaries: exponential family}

We focus on the Hamiltonians on $\Omega$ with a certain number of externally controllable parameters. We assume that the Hamiltonians have the form $\beta H_{\theta} (\omega) \coloneqq  - \theta \cdot  \phi (\omega) - h(\omega)$, where the parameter $\theta$ corresponds to the control parameters taking values in the parameter space $\Theta \subseteq \mathbb{R}^k$, $\phi : \Omega \to \mathbb{R}^k$ is the conjugate variable with respect to $\theta$, the dot `` $\cdot$ '' denotes the usual inner product in $\mathbb{R}^k$ and $h : \Omega \to \mathbb{R}$ is the remaining part of the Hamiltonian. $\Theta$ is assumed to be an open and convex subset of $\mathbb{R}^k$. Moreover, we suppose that the representation of the Hamiltonians is minimal, i.e., both $\{ \theta_i : i = 1, \dots, k \}$ and $\{ \phi_i : i = 1, \dots, k \}$ are affinely independent.

The Gibbs distribution for these Hamiltonians forms a minimal canonical \textit{exponential family} \cite{bn} $\mathcal{P}(\Theta, \phi, h)$, which is a family of distributions $P_{\theta}$ on $\Omega$ with densities $p_{\theta} (\omega) \coloneqq g_{\beta H_{\theta}} (\omega) = e^{\theta \cdot \phi (\omega) + h(\omega) - \psi (\theta)}$. The function $\psi (\theta) \coloneqq \ln \sum_{\omega \in \Omega} e^{\theta \cdot  \phi (\omega) + h(\omega)}$ is the Massieu function, which is related to the free energy $F_{\beta}(H_{\theta})$ through $F_{\beta} (H_{\theta}) = - \beta^{-1} \psi (\theta)$. The function $\psi (\theta)$ is differentiable infinitely many times and strictly convex on $\Theta$. The expectation value and covariance matrix of $\phi$ are respectively given by $\mathbb{E}_{\theta}(\phi) = \nabla_{\theta} \psi (\theta) \eqqcolon \mu (\theta)$ and $\mathrm{Cov}_{\theta}(\phi) = \nabla_{\theta} \nabla_{\theta} \psi (\theta) = \mathbb{E}_{\theta} [ - \nabla_{\theta} \nabla_{\theta} \ln P_{\theta}] \eqqcolon I(\theta)$. $I(\theta)$ is the Fisher information at $\theta$ for the family $\mathcal{P}(\Theta, \phi, h)$ and positive definite over $\Theta$. In the language of statistical physics, $I(\theta)$ is the isothermal susceptibility matrix and the relation $\mathrm{Cov}_{\theta} (\phi) = I(\theta)$ gives the fluctuation-response relation for a static isothermal response. The strict convexity of $\psi$ implies that the map $\theta \mapsto \mu (\theta)$ is invertible and that the elements of $\mathcal{P}(\Theta, \phi, h)$ can be reparametrized by the mapping $\mu \mapsto P_{\theta}$. $\theta (\mu)$ denotes the inverse and $\Xi \coloneqq \mu (\Theta)$ the conjugate parameter space. We note that $\nabla_{\mu} \theta (\mu) = I(\mu)$ gives the Fisher information in $\mu$-parametrization and is equals to $\mathrm{Cov}_{\mu}^{-1}(\phi)$. The Kullback-Leibler divergence between $P_{\theta}$ and $P_{\theta^{\prime}}$ is given by
\begin{align} \label{eq:kl}
 D (P_{\theta} \| P_{\theta^{\prime}}) = (\theta - \theta^{\prime}) \cdot  \mu (\theta) - \psi (\theta) + \psi (\theta^{\prime}).
\end{align}

The class of Hamiltonians having the above form describes a variety class of models that appears in statistical physics. We give two examples belonging to the above class.

\begin{exmp}[One-dimensional Ising model] \label{ex:ising}
The first example is a one-dimensional Ising model. Let $N$ be a positive integer. The state space of $N$ Ising spins on one-dimensional chain is $\Omega = \{ -1, 1 \}^N$ and a state is described by $\omega = (\sigma_1, \dots, \sigma_N)$, where $\sigma_i \in \{-1, 1\}$ denotes the state of spin at the site $i \in \{ 1, \dots, N \}$. The Ising model under a homogeneous magnetic field is described by the Hamiltonian $\beta H_{\theta} = - \beta h_{\mathrm{ex}} \sum_{i=1}^{N} \sigma_i - \beta J \sum_{i=1}^{N-1} \sigma_i \sigma_{i+1}$, where $h_{\mathrm{ex}}$ is the external field and $J$ is the coupling constant. If the agent changes only the magnetic field uniformly, the parameter is chosen as $\theta = \beta h_{\mathrm{ex}} \in \Theta \coloneqq \mathbb{R}$, the conjugate variable as the total magnetization $\phi (\omega) = \sum_{i=1}^{N} \sigma_i$, and the remaining part as $h(\omega) = \beta J \sum_{i=1}^{N-1} \sigma_i \sigma_{i+1}$. In this case, the expectation of $\phi$ is the average magnetization $\mu (\theta) = \sum_{i=1}^{n} P_{\theta} (\sigma_i)$ and the Fisher information is $I(\theta) = \beta^{-1} \chi$, where $\chi$ is the magnetic susceptibility. If the agent changes the coupling constant $J$ in addition to the magnetic field, the parameter becomes two-dimensional vector $\theta = (\beta h_{\mathrm{ex}}, \beta J)$, the conjugate variable $\phi = (\sum_{i=1}^{n} \sigma_i, \sum_{i=1}^{n-1} \sigma_i \sigma_{i+1})$ and the remaining part $h = 0$. The parametrization for a given physical model is thus not unique in general.
\end{exmp}

\begin{exmp}[Positive distribution] \label{ex:pos}
The second example is positive distributions on a general finite state space. This case corresponds to that we argued in Section \ref{sec:gh}. Let $\Omega = \{0, 1,2,\dots, k\}$ be a finite set. Putting $\Theta = \{ ( \ln (p_i / p_0))_{i=1}^{k} \in \mathbb{R}^{k} : 0 < p_i < 1, \sum_{j=1}^{k} p_j < 1, p_0 = 1 - \sum_{i=1}^{k} p_i \}$, $\theta_i = \ln ( p_i / p_0)$, $\phi_i = 1_{\{ i \}}$, $h = 0$, $\psi (\theta) = - \ln (1 - \sum_{i=1}^{k} p_i) = - \ln p_0$, we have that
\begin{align}
 p_{\theta}(i) = \exp \left\{ \sum_{j=1}^{k} \theta_j \phi_j (i) - \psi (\theta) \right\} = p_i
\end{align}
for every $i \in \Omega$. The exponential family is therefore the family of strictly positive distributions on $\Omega$. The empirical mean of $\phi$ is identified with the empirical density and the expectation $\mu_i (\theta)$ is given by $p_i$ for $i = 1,\dots, k$.
\end{exmp}

\subsection{Setup and Result}

The protocol we study in this subsection is almost the same as that in Section \ref{sec:gh}. A crucial difference is that the agent cannot prepare an arbitrary Hamiltonian in general and he has to construct a strategy by tuning only a small number of control parameters. 

We assume that the initial Hamiltonian has the form $\beta H_{\theta} = - \theta \cdot \phi - h$ for some $\theta \in \Theta$. A strategy the agent applies is characterized by a function $\hat{\theta} : \Omega^* \to \Theta$. The operation in the $n$-th cycle is performed according to the following protocol:
\begin{itemize}
\setlength{\leftskip}{3.0mm}
 \item[(P1)] The agent quenches the Hamiltonian adiabatically from the initial Hamiltonian $H_{\theta}$ to another one $H_{\hat{\theta}(\omega^{n-1})}$ when the initial states up to the $(n-1)$-th cycle are $\omega^{n-1} = \omega_1 \dots \omega_{n-1}$.
 \item[(P2)] The agent equilibrates the system with the inverse temperature $\beta$.
 \item[(P3)] The agent resets the system quasi-statically and isothermally.
\end{itemize}

The accumulation of extracted work $W^{\hat{\theta}} \coloneqq W^{H_{\hat{\theta}}}$ in this protocol under the strategy $\hat{\theta}$ is given by
\begin{align} \label{eq:pew}
 \beta W^{\hat{\theta}}(\omega^n) = \sum_{i=1}^{n} & \left[ (\hat{\theta}(\omega^{i-1}) - \theta) \cdot  \phi (\omega_i) - \psi (\hat{\theta}(\omega^{i-1})) + \psi (\theta) \right].
\end{align}
Let $E \subseteq \Omega^{\mathbb{N}_{+}}$ be an almost sure property under the infinite product of the Gibbs distributions $p_{\theta}^{\mathbb{N}_{+}} = g_{\beta H_{\theta}}^{\mathbb{N}_{+}}$. Although there exists a strategy $\hat{H}$ that forces $E$ according to Ville's theorem (Theorem \ref{thm:ville}), the strategy $\hat{H}$ may not be realizable within the above protocol. In general, the decrease in the number of possible strategies the agent can apply leads to the decrease in the variety of almost sure properties forced by the second law. The question we study in this section is what statistical properties are forced when we restrict the ability of the agent. 

We expect from the expression (\ref{eq:pew}) that such properties are restricted to the statistics of the conjugate variable $\phi$ with respect to the control parameter $\theta$. Indeed, we find that there exists a strategy in the restricted protocol that forces the strong law of large numbers for the conjugate variable.

\begin{thm} \label{thm:p}
Let $\overline{\phi}_n$ be the empirical mean of $\phi$ defined by
\begin{align}
 \overline{\phi}_n (\omega^n) \coloneqq \frac{1}{n} \sum_{i=1}^{n} \phi (\omega_i).
\end{align}
There exists a strategy $\hat{\theta}: \Omega^* \to \Theta$ that forces
\begin{align} \label{eq:pc}
 \lim_{n \to \infty} \overline{\phi}_n  = \mu (\theta).
\end{align}
\end{thm}

This is the second main result of the present paper. Theorem \ref{thm:np} is a special case of Theorem \ref{thm:p} because the empirical mean of $\phi$ in Example \ref{ex:pos} can be identified with the empirical density.

Since $\mu (\theta)$ is the equilibrium value of the conjugate variable, Theorem \ref{thm:p} implies that the infinite sequences satisfying the second law are indistinguishable from random sequences sampled from the equilibrium distribution as long as we observe the conjugate variable. In contrast to Theorem \ref{thm:np}, the relative frequencies of initial microscopic states may not converge to the Gibbs distribution under the strategy. For instance, when the agent controls the homogeneous magnetic field on the free spin system, the average magnetization should converge to the equilibrium value but the relative frequencies of spin configurations having the same magnetization may not be controlled. Thus the statistical properties the agent can force under the second law depend on what kind of physical operations are allowed for the agent.

\subsection{Sketch of proof of Theorem \ref{thm:p}} \label{subsec:proof}

We sketch a proof of Theorem \ref{thm:p}. We present the details of the proof in Appendix \ref{ap:proof}. 

By the following lemma, it is sufficient to prove the existence of a strategy that weakly forces (\ref{eq:pc}) and takes values in a compact set containing the initial parameter $\theta$ as an interior point. The proof is presented in Appendix \ref{ap:proof}.

\begin{lem} \label{lem:wff}
Let $E \subseteq \Omega^{\mathbb{N}}$ be an event and $\Theta_0 \subset \Theta$ a compact convex subset with $\theta \in \mathrm{int} \Theta_0$. If there exists a strategy that weakly forces $E$ and takes values in $\Theta_0$, then there is also a strategy that forces $E$ and takes values in the same set $\Theta_0$.
\end{lem}

Our method of the proof is based on the maximum likelihood strategy. If the agent applies a constant strategy, i.e., $\hat{\theta} = \theta^{\prime}$, the accumulation of extracted work up to the $n$-th cycle is given by
\begin{align} \label{eq:const}
 \beta W^{\hat{\theta}}(\omega^n) = \ln \frac{p^n_{\theta^{\prime}}(\omega^n)}{p^n_{\theta} (\omega^n)}.
\end{align}
The maximum likelihood strategy is defined as the maximizer of (\ref{eq:const}) for each $\omega^n$. Because the maximizer of (\ref{eq:const}) over the open set $\Theta$ may not exist, we restrict the range of the strategy to a compact set $\Theta_0$ containing the initial parameter $\theta$ as an interior point. The maximum likelihood estimator $\hat{\theta}_{\mathrm{ML}} : \Omega^* \to \Theta_0$ with respect to $\Theta_0$ is defined as
\begin{align} \label{eq:defml}
 \hat{\theta}_{\mathrm{ML}} (\omega^n) \coloneqq  \arg \max_{\theta^{\prime} \in \Theta_0} \ln p_{\theta^{\prime}}^{n}(\omega^n).
\end{align}
We define the corresponding conjugate parameter space $\Xi_0 \coloneqq \mu (\Theta_0)$, which is also a compact convex subset of $\Xi$. We set $\hat{\mu}_{\mathrm{ML}}(\omega^n) = \mu (\hat{\theta}_{\mathrm{ML}}(\omega^n))$. Because $\nabla_{\theta} \ln p_{\theta}(\omega) = \phi (\omega) - \mu (\theta)$, the maximum likelihood estimator $\hat{\theta}_{\mathrm{ML}}$ and $\hat{\mu}_{\mathrm{ML}}$ are given as the solution of the equation
\begin{align} \label{eq:ml}
 \frac{1}{n} \sum_{i=1}^n \phi (\omega_i) = P_{\hat{\theta}_{\mathrm{ML}}(\omega^n)}(\phi) = \hat{\mu}_{\mathrm{ML}}(\omega^n)
\end{align}
if and only if $\overline{\phi}_n \in \Xi_0$, which does not always hold. 

The maximum likelihood strategy was originally studied in the context of universal coding and sequential prediction. Kot\l owski and Gr\"unwald obtained the following lower bound in \cite{kg}:
\begin{align}
&\sum_{i=1}^{n} \left[ \ln p_{\hat{\theta}_{\mathrm{ML}}(\omega^{i-1})}(\omega_i) - \ln p_{\hat{\theta}_{\mathrm{ML}}(\omega^n)} (\omega_i) \right] \geq  - \frac{I_{\Xi_0} (B + C_{\Xi_0})^2}{2} \ln n + O(1),
\end{align}
where $B \coloneqq \max_{\omega \in \Omega} | \phi (\omega) |$, $C_{\Xi_0} \coloneqq \max_{\mu \in \Xi_0} \| \mu \|$ and $I_{\Xi_0} \coloneqq \max_{\mu \in \Xi_0} \| I(\mu) \|$. We note that these constants are all finite due to the compactness of $\Omega$ and $\Xi_0$. By using this result, we have a lower bound for the accumulation of the extracted work under the maximum likelihood strategy:
\begin{align} \label{eq:bound}
 \beta W^{\hat{\theta}_{\mathrm{ML}}} (\omega^n) &\geq n \left[ D (P_{\hat{\theta}_{\mathrm{ML}}(\omega^n)} \| P_{\theta})  - \frac{I_{\Xi_0} (B + C_{\Xi_0})^2}{2} \frac{\ln n}{n} \right] + O(1).
\end{align}
The details of the calculation is presented in Appendix \ref{ap:proof}. We remark that the lower bound (\ref{eq:bound}) is valid for any individual sequences $\omega^n$. The inequality (\ref{eq:bound}) implies that the maximum likelihood strategy forces weakly the convergence of $p_{\hat{\theta}_{\mathrm{ML}}} (\omega)$ to $p_{\theta} (\omega)$ for any $\omega \in \Omega$, and therefore it also forces weakly $\overline{\phi}_n \to \mu (\theta)$. This completes the proof.

We have several remarks before ending the subsection. 
\begin{itemize}
 \item[(1)] The lower bound (\ref{eq:bound}) has much information on the asymptotic behavior. If the second law is weakly valid, i.e., $\sup_n W^{\hat{\theta}_{\mathrm{ML}}} < \infty$, then $P_{\hat{\theta}_{\mathrm{ML}}}$ converges to $P_{\theta}$ with respect to the Kullback-Leibler distance with the convergence rate $O(\sqrt{ \ln n / n})$ and the convergence factor $\sqrt{ I_{\Xi_0} (B + C_{\Xi_0})^2 / 2}$. We also note that if the convergence (\ref{eq:pc}) does not hold, the accumulation of the extracted work grows at least linearly. Specifically, the extracted work \textit{per cycle} becomes positive infinitely many times,
\begin{align} \label{eq:slpr}
 \limsup_{n \to \infty} \frac{W^{\hat{\theta}}_n}{n}  > 0,
\end{align}
for sequences violating the strong law of large numbers (\ref{eq:pc}). For instance, if $\hat{\theta}_{\mathrm{ML}}(\omega^n) \to \theta^{\prime} ( \neq \theta)$ as $n \to \infty$, the rate of the extracted work is given by the Kullback-Leibler distance between $P_{\theta^{\prime}}$ and $P_{\theta}$. This implies that our results (Theorem \ref{thm:np} and \ref{thm:p}) are valid even if we apply (\ref{eq:slpr}) as a definition of the violation of the second law.
\item[(2)] The maximum likelihood estimator $\hat{\theta}_{\mathrm{ML}} (\omega_n)$ defined by (\ref{eq:defml}) depends on $\omega_n$ through only $\overline{\phi}_n$ because  the conjugate variable $\phi$ is the sufficient statistic for the exponential family $\mathcal{P} (\Theta, \phi, h)$. Therefore, even if only the information on the past history of $\phi$ is given to the agent, he can perform the maximum likelihood strategy and can forces the strong law of large numbers for the conjugate variable.
\end{itemize}

\subsection{Numerical Demonstration}

We illustrate the validity of the maximum likelihood strategy numerically for the Ising Hamiltonian under the homogeneous magnetic field for two spins, $\beta H_{\theta} (\sigma_1, \sigma_2) = - \sigma_1 \sigma_2 - \theta ( \sigma_1 + \sigma_2)$. The initial parameter is set to be $\theta = 0$. In Fig \ref{fig:work}, we plot the accumulations of extracted work for sequences generated from the Gibbs distribution with $\theta = 0, \pm 0.1, \pm 0.2$. The figure shows that while the accumulations for parameters $\theta = \pm 0.1, \pm 0.2$ diverge to infinity and the second law is violated, the accumulation for $\theta = 0$ remains finite. The linear growth of the accumulation of extracted work is consistent with the lower bound (\ref{eq:bound}).

\begin{figure}
\centering
\includegraphics[width=10cm]{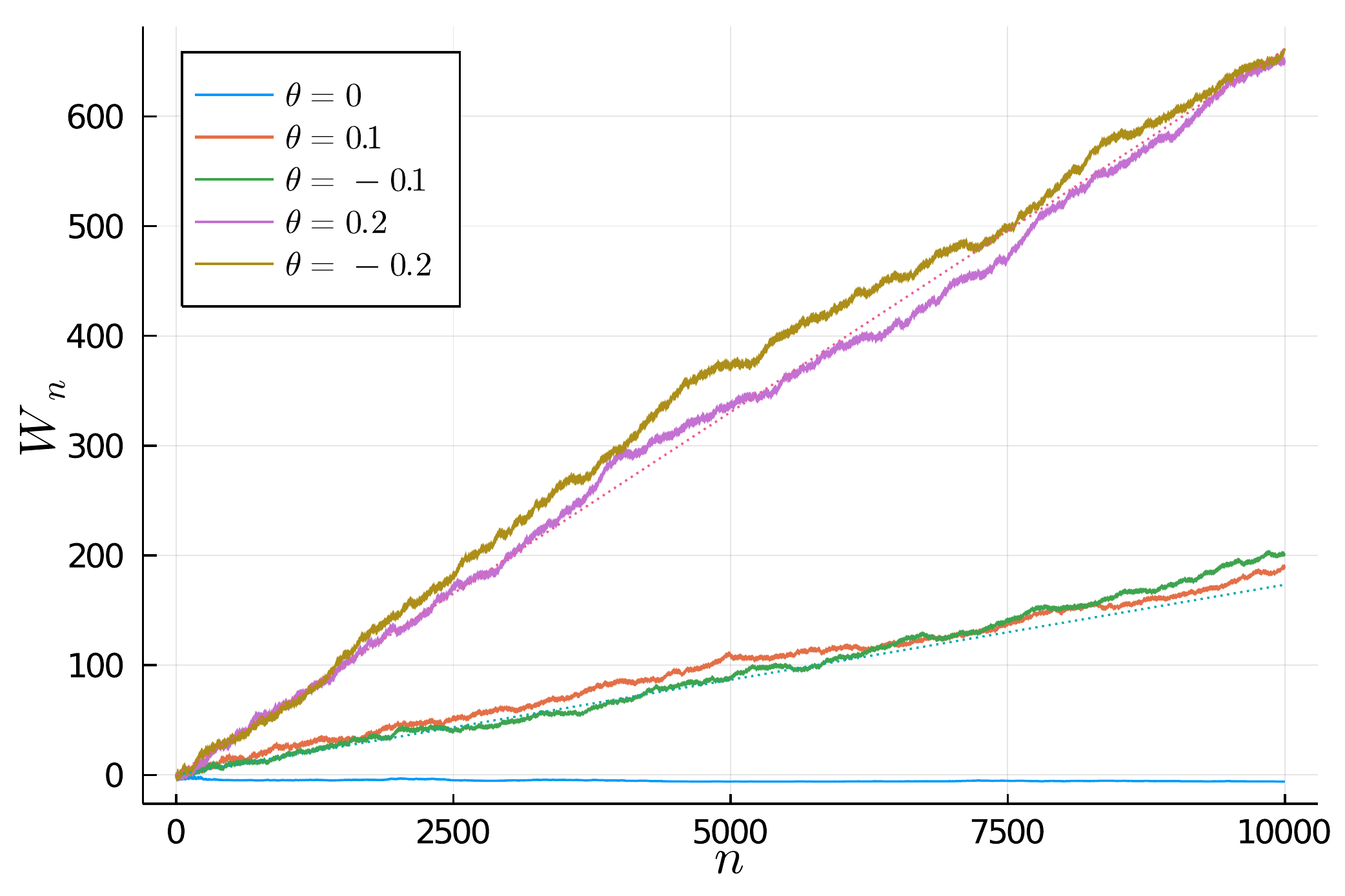}
\caption{The accumulations of the extracted work for sequences generated by the Gibbs distributions with various parameters versus the number of cycles. The top (resp. bottom) dotted line indicates the Kullback-Leibler distance between $P_{\theta = \pm 0.2}$ (resp. $P_{\theta = \pm 0.1}$) and $P_{\theta = 0}$.}
\label{fig:work}
\end{figure}

\section{Game-theoretic Interpretation}
\label{sec:gi}

We discuss a game-theoretic interpretation of the protocol of this paper. The interpretation is based on the analogy with game-theoretic probability theory \cite{sv1,sv2}. 

We interpret the work extraction as a game played between two players, $\Skeptic$ and $\Reality$. The player $\Skeptic$ doubts that the statistical property of the heat engine is described by equilibrium statistical mechanics, particularly the Gibbs distribution. He attempts to test the hypothesis of equilibrium statistical mechanics by actually extracting work from the engine many times and observing whether the second law of thermodynamics holds or not. If the accumulation of the extracted work diverges to infinity, i.e., the second law is violated, then $\Skeptic$ rejects the hypothesis. On the other hand, the player $\Reality$ decides initial microscopic states of the heat engine. The duty of $\Reality$ is to prevent $\Skeptic$ from extracting the infinite amount of work and to make the second law of thermodynamics valid in our world.

The algorithm of the game corresponding to the protocol we discussed in Section \ref{sec:gh} is as follows.

\begin{protocol}[Generic Hamiltonian] \label{pro:np}
\

\textbf{Parameter:}

 \ \ \ an initial Hamiltonian $H : \Omega \to \mathbb{R}$,

 \ \ \ an inverse temperature $\beta > 0$.

\textbf{Players:} $\Skeptic$, $\Reality$.

\textbf{Protocol:}

 \ \ \ $W (\square) = 0$.

 \ \ \ FOR $n=1,2, \dots$:

\vspace{-1.5mm}

\begin{itemize}
 \item[] \ \ \ \ \ \   $\Skeptic$ announces $H_n$
 \item[] \ \ \ \ \ \ $\Reality$ announces $\omega_n \in \Omega$
 \item[] \ \ \ \ \ \ $W_n \coloneqq W_{n-1} + H(\omega_n) - H_n(\omega_n) + F_{\beta}(H_n) - F_{\beta}(H)$
\end{itemize}

\vspace{-1.5mm}

 \ \ \ END FOR  
\end{protocol}

In Protocol \ref{pro:np}, $H_n \coloneqq \hat{H}(\cdot | \omega^{n-1})$ denotes the quenched Hamiltonian in the $n$-th cycle and $W_n =W^{\hat{H}}(\omega^n)$ the accumulation of extracted work up to the $n$-th cycle. We remark that $\Reality$ can decides an initial microscopic state $\omega_n$ \textit{after} the announcement by $\Skeptic$. Therefore, we allow $\Reality$ to move strategically in Protocol \ref{pro:np}. Nevertheless, Theorem \ref{thm:np} is still valid even if $\Reality$ decides microscopic states strategically because whether a strategy forces a property or not is independent of the strategy $\Reality$ applies. Theorem \ref{thm:np} implies that there exists a strategy for $\Skeptic$ such that $\Reality$ has to converge the empirical distribution to the Gibbs distribution. In other words, $\Reality$ is forced to act probabilistically in a manner consistent with equilibrium statistical mechanics due to the second law. 

As an application of Theorem \ref{thm:np}, we consider a testing procedure of random number generator. Let us suppose that $\Reality$ claims that she finds an algorithm to generate random numbers with respect to the Gibbs distribution. $\Skeptic$ doubts her claim and attempts to test it. Although there are several criteria that the random numbers should satisfy, we consider the strong law of large numbers (\ref{eq:gslln}) for the empirical distribution as a criterion here. If the random number obtained from $\Reality$'s generator does not satisfy the SLLN, $\Skeptic$ rejects her claim and concludes that the random number generator does not work well. According to Theorem \ref{thm:p}, $\Skeptic$ can confirm the validity of the SLLN by performing cyclic operations for microscopic states prepared by the generator. In other words, the second law of thermodynamics can be used as a test for the random number generators.

Similarly, the algorithm of the protocol in Section \ref{sec:ph} is as follows.

\begin{protocol}[Parametric Hamiltonian] \label{pro:p}
\

\textbf{Parameter:}

 \ \ \ a control parameter space $\Theta$,
 
 \ \ \ a conjugate variable $\phi : \Omega \to \mathbb{R}^k$,
 
 \ \ \ a remaining Hamiltonian $h : \Omega \to \mathbb{R}$,

 \ \ \ an initial parameter $\theta \in \Theta$, 

 \ \ \ an inverse temperature $\beta > 0$.

\textbf{Players:} $\Skeptic$, $\Reality$.

\textbf{Protocol:}

 \ \ \ $W (\square) = 0$.

 \ \ \ FOR $n=1,2, \dots$:

\vspace{-1.5mm}

\begin{itemize}
 \item[] \ \ \ \ \ \ $\Skeptic$ announces $\theta_n \in \Theta$.
 \item[] \ \ \ \ \ \ $\Reality$ announces $\omega_n \in \Omega$
 \item[] \ \ \ \ \ \ $W_n \coloneqq W_{n-1} + H_{\theta}(\omega_n) - H_{\theta_n}(\omega_n) + F_{\beta}(H_{\theta_n}) - F_{\beta}(H_{\theta})$
\end{itemize}

\vspace{-1.5mm}

END FOR  
\end{protocol}

Theorem \ref{thm:p} can be interpreted game-theoretically in a similar way to the case of Theorem \ref{thm:np}.

\section{Concluding Remarks}
\label{sec:cr}

In this paper, we provided a novel formulation of the second law of thermodynamics and showed that there exist strategies for the agent that force statistical properties consistent with equilibrium statistical mechanics, i.e., the strong law of large numbers for the empirical distribution and empirical mean. In the protocol where the agent is able to prepare arbitrary Hamiltonians, the empirical distribution must converge to the Gibbs distribution for the initial Hamiltonian under some strategy. In the protocol where the agent can change a small number of parameters in the Hamiltonian, the maximum likelihood strategy forces the strong law of large numbers for the conjugate variable. Before ending the paper, we discuss future directions of the study.

First, we considered the simple settings where initial microscopic states are prepared independently and identically. It is important to study in the same spirit more complicated settings such as stochastic thermodynamics \cite{seifert,sekimoto}, information thermodynamics \cite{su}, and quantum systems \cite{pw,lenard}. In particular, it is a future subject to treat the stochastic evolution of microscopic state in our framework and to connect it with thermodynamic quantities. The extension to the quantum settings might be more difficult  because the quantum theory has a probabilistic structure different from the classical one. See the chapter 10.6 in \cite{sv2} for game-theoretic formulations for Born's rule and quantum computation.

Secondly, we assume in our settings that the work extracted in the quasi-static isothermal process is equal to the decrease in free energy defined by $F_{\beta}(H) = - \beta^{-1} \ln \sum_{\omega \in \Omega} e^{- \beta H(\omega)}$. This assumption allows us to connect the work extracted in the quasi-static isothermal process with the microscopic Hamiltonian, without explicitly referring to the time evolution of the system. However, we should remove it to clarify the emergence of the probabilistic structure from thermodynamics. Although one way to do so is to replace the difference in free energy, $F_{\beta}(\hat{H}) - F_{\beta}(H)$, in Protocol \ref{pro:np} by the work in the quasi-static process defined only in terms of mechanics, it is a challenging task because we have to take the dynamical properties of the system into account explicitly.

Thirdly, we investigated the asymptotic behavior of the work extracted from a \textit{finite} system as the number of cycles $n$ goes to infinity. While our results are valid for small systems as with stochastic thermodynamics, the second law of thermodynamics is believed to hold almost surely for \textit{macroscopic} systems without repeating the operations. Developing our formulation for macroscopic thermodynamics is a future direction of the study.

Fourthly, we assumed in our protocol that the agent can use the information about microscopic states in the past. However, this assumption may be too demanding when the size of the system is large. In a more realistic situation, the agent can measure only the value of extracted work. It is important to study statistical properties forced by strategies under such situation.

Fifthly, the approach based on the prequential form allows us to classify the statistical property of the equilibrium state according to the ability of the agent. In the usual formulation of statistical mechanics, the law of large numbers for the empirical distribution and for the empirical mean of the conjugate variable are both formulated as the events with probability one. In this sense, there is no distinction between these probabilistic laws. However, in the prequential approach, what kind of probabilistic laws can be observed depends on our ability. It is interesting to provide a detailed classification of the probabilistic laws of equilibrium states according to the ability of the agent to operate the system.

Finally, we mention the probabilistic feature of infinite sequences. Although the empirical distribution for the simple binary sequence $01010101 \dots$ converges to the uniform distribution on $\{0, 1\}$, it is not sufficiently random according to our intuition. However, Theorems \ref{thm:np} and \ref{thm:p} refer to only the empirical distribution and mean. To say that an infinite sequence is random with respect to a probability distribution, we have to require stronger conditions on the empirical statistics. The algorithmic theory of randomness provides an idealized notion of randomness on the basis of computability theory and martingale theory \cite{nies,dh}. A class of randomness is defined as the intersection of sets that are weakly forced by strategies with some computability condition. For instance, computable strategies specify the class of randomness called \textit{computable random}. Although there exists no universal computable strategy, i.e., the computable randomness cannot be forced by a single strategy, to investigate such class in the context of thermodynamics may be interesting. We hope that the present paper provides an insight into studying thermodynamics from the viewpoints of the martingale structure and computability of strategies. As a study in the same spirit, see \cite{hs} showing the relevance of algorithmic randomness to thermodynamic irreversibility.

\begin{acknowledgements}
The author thanks Shin-ichi Sasa for carefully reading the manuscript and making valuable comments. The present work was supported by JSPS KAKENHI Grant Number JP20J12143.
\end{acknowledgements}

\appendix

\section{Proof of Eq. (\ref{eq:pv2l})}
\label{ap:pv2l}

We prove a stronger statement $P \{ \sup_n W_n^{\hat{\mu}} = \infty \} = 0$. First, we note that $P \{ \sup_n W_n^{\hat{\mu}} = \infty \}  = P \{ \sup_{n} e^{\beta W_n^{\hat{\mu}}} = \infty \}$. $e^{\beta W_n^{\hat{\mu}}}$ is a positive martingale and bounded in $L^1$ because $\mathbb{E}_P (e^{\beta W_n^{\hat{\mu}}}) = 1$ for any $n$. We have from the martingale convergence theorem that $e^{\beta W_{\infty}^{\hat{\mu}}} \coloneqq \lim_{n \to \infty} e^{\beta W_n^{\hat{\mu}}}$ exists and is a finite non-negative value almost surely. Hence, $P \{ \sup_n W_n^{\hat{\mu}} = \infty \} = 0$.

\section{Exponentiated Extracted Work = Martingale}
\label{ap:mg}

For a given strategy $\hat{H}: \Omega^* \to \mathbb{R}$, the exponential of the accumulation of extracted work $e^{\beta W^{\hat{H}}}$ is a positive martingale with respect to $g_{\beta H}^{\mathbb{N}_{+}}$. In fact, 
\begin{align*}
 \mathbb{E}_{g_{\beta H}} [ e^{\beta W_n^{\hat{H}}} | \omega^{n-1} ] &= \sum_{\omega_n \in \Omega} \frac{\hat{q}(\omega^{n})}{g_{\beta H}^{n} (\omega^n)} g_{\beta H}(\omega_n) = \sum_{\omega_n \in \Omega} \frac{\hat{q}(\omega^{n-1} \omega_n)}{g_{\beta H}^{n-1} (\omega^{n-1})} = \frac{\hat{q} (\omega^{n-1})}{g_{\beta H}^{n-1}(\omega^{n-1})} = e^{\beta W^{\hat{H}}(\omega^{n-1})},
\end{align*}
where we have used $\sum_{\omega_n \in \Omega} \hat{q}(\omega^{n-1} \omega_n) = \hat{q}(\omega^{n-1})$ in the last line.

Conversely, suppose that a process $M : \Omega^* \to \mathbb{R}_{+}$ is a positive martingale with respect to $g_{\beta H}^{\mathbb{N}_{+}}$ and $M(\square) = 1$. We then have that $\hat{q} (\omega_n | \omega^{n-1}) \coloneqq g_{\beta H}(\omega_n) M(\omega^{n-1} \omega_n) / M(\omega^{n-1})$ is a positive probability density on $\Omega$. In fact, the positivity of $M$ leads to $\hat{q}(\omega_n | \omega^{n-1}) > 0$ and the martingale property of $M$ implies that
\begin{align*}
 \sum_{\omega_n \in \Omega} \hat{q}(\omega_n | \omega^{n-1}) &= \frac{1}{M(\omega^{n-1})} \sum_{\omega_n \in \Omega} g_{\beta H}(\omega_n) M(\omega^{n-1} \omega_n) = \frac{M(\omega^{n-1})}{M(\omega^{n-1})} = 1.
\end{align*}
By defining the strategy $\hat{H}$ through the relation $g_{\beta \hat{H}(\cdot | \omega^{n-1})} (\omega_n) = \hat{q}(\omega_n | \omega^{n-1})$, we have that $e^{\beta W^{\hat{H}}(\omega^n)} = M(\omega^n)$ for any $\omega^n \in \Omega^*$ including the case $\omega^n = \square$.

\section{Proof of Theorem \ref{thm:p}}
\label{ap:proof}

We present a complete proof of Theorem \ref{thm:p}. 

First, we prove Lemma \ref{lem:wff}. It is proved in the same way as Lemma 3.1 in \cite{sv1}.

\begin{proof}[Proof of Lemma \ref{lem:wff}]
For a given strategy $\hat{\theta} : \Omega^* \to \Theta_0$ and a positive real number $C>0$, we define the ``stopped strategy'' $\hat{\theta}^{(C)}$ as
\begin{align}
 \hat{\theta}^{(C)} (\omega^n) =
 \begin{cases}
  \hat{\theta}(\omega^n) & \text{if $\beta W^{\hat{\theta}}(\omega^m) < C$ for all $m \leq n$} \\
  \theta & \text{otherwise}.
 \end{cases}
\end{align}
Let $\hat{\theta}$ be a strategy that weakly forces $E$ and $\hat{\theta}^{(C)}$ its stopped strategy for $C>0$. Consider the countable number of stopped strategies $\hat{\theta}^{(2^i)}$ for $i=1,2,\dots$. Since $\hat{\theta}^{(2^i)} (\omega^n) \in \Theta_0$ and $\Theta_0$ is compact, $\hat{\theta}^* (\omega^n) \coloneqq \sum_{i=1}^{\infty} 2^{-i} \hat{\theta}^{(2^i)} (\omega^n)$ exists for any $\omega^n \in \Omega^*$. The closedness and convexity of $\Theta_0$ lead to $\hat{\theta}^*(\omega^n) \in \Theta_0$. Hence, the function $\hat{\theta}^* : \Omega^* \to \Theta_0$ defines a strategy taking values in $\Theta_0$. We obtain from the convexity of $\psi$ and Jensen's inequality that
\begin{align}
 W^{\hat{\theta}^*}(\omega^n) \geq \sum_{i=1}^{\infty} 2^{-i} W^{\hat{\theta}^{(2^i)}}(\omega^n).
\end{align}
For $\omega^{\infty} \in E^c$, the limit $\lim_n \beta W^{\hat{\theta}^{(C)}} (\omega^n)$ exists and is larger than or equal to $C$ because $\sup_n W^{\hat{\theta}}(\omega^n) = \infty$. Since $\lim_n \beta W^{\hat{\theta}^{(2^i)}}(\omega^n) \geq 2^i$ for $\omega^{\infty} \in E^c$, $\lim_{n} W^{\hat{\theta}^*}(\omega^n)$ diverges to infinity for $\omega^{\infty} \in E^c$. 
\end{proof}

Next, we summarize the basic properties of the maximum likelihood estimator. The convexity of $- \ln p^n_{\theta} (\omega^n)$ with respect to $\theta$ implies that
\begin{align} \label{eq:pd}
 (\theta - \hat{\theta}_{\mathrm{ML}}(\omega^n)) \cdot  (\hat{\mu}_{\mathrm{ML}}(\omega^n) - \overline{\phi}_n ) \geq 0
\end{align}
for all $\theta \in \Theta_0$. In fact, assume that it does not hold for some $\theta \in \Theta_0$. Consider the continuous path $s(t) \coloneqq t \theta + (1-t) \hat{\theta}_{\mathrm{ML}}(\omega^n)$ for $t \in [0,1]$. The convexity of $\Theta_0$ implies that $s(t) \in \Theta_0$ for all $t \in [0,1]$. From the assumption, we obtain that $\partial_t [ - \ln p^n_{s(t)} (\omega^n)] |_{t=0} = n ( \hat{\mu}_{\mathrm{ML}}(\omega^n) - \overline{\phi}(\omega^n)) \cdot (\theta - \hat{\theta}_{\mathrm{ML}}(\omega^n)) < 0$, and thus $- \ln p^n_{s(t)} (\omega^n) < - \ln p^n_{\hat{\theta}_{\mathrm{ML}}(\omega^n)} (\omega^n)$ for a sufficiently small $t \in [0, 1]$. This contradicts the minimality of $\hat{\theta}_{\mathrm{ML}}(\omega^n)$ on $\Theta_0$. We have from the condition (\ref{eq:pd}) that if $\hat{\theta}_{\mathrm{ML}}(\omega^n) \in \mathrm{int} \Theta_0$, then $\overline{\phi}_n \in \Xi_0$ and therefore $\hat{\mu}_{\mathrm{ML}}(\omega^n) = \overline{\phi}_n$. Otherwise, $(\hat{\mu}_{\mathrm{ML}}(\omega^n) - \overline{\phi}_n) \cdot (\theta - \hat{\theta}_{\mathrm{ML}}(\omega^n)) < 0$ for some $\theta \in \Theta_0$ in a neighborhood of $\hat{\theta}_{\mathrm{ML}}(\omega^n)$, which contradicts (\ref{eq:pd}). 

The following lower bound proved by Kot\l owski and Gr\"unwald \cite{kg} is crucial for our proof.

\begin{lem}[\cite{kg}] \label{lem:kg}
\begin{align}
&\sum_{i=1}^{n} \left[ \ln p_{\hat{\theta}_{\mathrm{ML}}(\omega^{i-1})}(\omega_i) - \ln p_{\hat{\theta}_{\mathrm{ML}}(\omega^n)} (\omega_i) \right] \geq - \frac{I_{\Xi_0} (B + C_{\Xi_0})^2}{2} \ln n + O(1),
\end{align}
where $B \coloneqq \max_{\omega \in \Omega} | \phi (\omega) |$, $C_{\Xi_0} \coloneqq \max_{\mu \in \Xi_0} \| \mu \|$ and $I_{\Xi_0} \coloneqq \max_{\mu \in \Xi_0} \| I(\mu) \|$. 
\end{lem}

By using Lemma \ref{lem:kg}, we obtain Theorem \ref{thm:p} immediately.

\begin{proof}[Proof of Theorem \ref{thm:p}]
For any strategy $\hat{\theta}$, $\beta W_n^{\hat{\theta}}$ can be decomposed into two parts:
\begin{align}
 \beta W^{\hat{\theta}} (\omega^n) &= \sum_{i = 1}^{n} \left[ \ln p_{\hat{\theta}(\omega^n)}(\omega_i) - \ln p_{\theta} (\omega_i) \right] + \sum_{i=1}^{n} \left[ \ln p_{\hat{\theta}(\omega^{i-1})}(\omega_i) - \ln p_{\hat{\theta}(\omega^n)} (\omega_i) \right].
\end{align}
First, we investigate the first part for $\hat{\theta}_{\mathrm{ML}}$. We have from (\ref{eq:kl}) and (\ref{eq:pd}) that
\begin{align} \label{eq:ft}
 \sum_{i = 1}^{n} \left[ \ln p_{\hat{\theta}_{\mathrm{ML}}(\omega^n)}(\omega_i) - \ln p_{\theta} (\omega_i) \right] &= n D(P_{\hat{\theta}_{\mathrm{ML}}(\omega^n)} \| P_{\theta}) + n (\hat{\theta}_{\mathrm{ML}}(\omega^n) - \theta) \cdot  ( \overline{\phi}_n - \hat{\mu}_{\mathrm{ML}}(\omega^n))
 \notag \\
 & \geq n D(P_{\hat{\theta}_{\mathrm{ML}}(\omega^n)} \| P_{\theta}).
\end{align}
Combining Lemma \ref{lem:kg} with the above inequality (\ref{eq:ft}), we have that
\begin{align} \label{eq:bound1}
 \beta W^{\hat{\theta}_{\mathrm{ML}}} (\omega^n) &\geq n \left[ D(P_{\hat{\theta}_{\mathrm{ML}}(\omega^n)} \| P_{\theta}) - \frac{I_{\Xi_0} (B + C_{\Xi_0})^2}{2} \frac{\ln n}{n} \right] + O(1).
\end{align}
Suppose that $\sup_n W^{\hat{\theta}_{\mathrm{ML}}}_n (\xi) < \infty$ for $\xi = \omega_1 \omega_2 \dots$. Then, there exists a real number $C_{\xi} \in \mathbb{R}$ such that $W^{\hat{\theta}_{\mathrm{ML}}} (\omega^n) \leq C_{\xi}$ for all $n \in \mathbb{N}_{+}$. From the inequality (\ref{eq:bound1}), we have that
\begin{align}
 D(P_{\hat{\theta}_{\mathrm{ML}}(\omega^n)} \| P_{\theta}) \leq \frac{C_{\xi}}{n} + \frac{I_{\Xi_0} (B + C_{\Xi_0})^2}{2} \frac{\ln n}{n}  + O(n^{-1}).
\end{align}
Since $C_{\xi}$ is independent of $n$, we obtain that $D(P_{\hat{\theta}_{\mathrm{ML}}(\omega^n)} \| P_{\theta}) \to 0$ as $n \to \infty$. This implies that $p_{\hat{\theta}_{\mathrm{ML}}(\omega^n)}(\omega)$ converges to $p_{\theta}(\omega)$ for any $\omega \in \Omega$, and therefore $\hat{\mu}_{\mathrm{ML}} (\omega^n) \to \mu (\theta)$. For $\xi \in \Omega^{\mathbb{N}_{+}}$ such that $\hat{\mu}_{\mathrm{ML}} \to \mu (\theta) \in \mathrm{int} \Xi_0$, $\overline{\phi}_{n} \in \Xi_0$ and $\hat{\mu}_{\mathrm{ML}}(\omega^n) = \overline{\phi}_n$ for sufficiently large $n$, and therefore $\overline{\phi}_n \to \mu (\theta)$.
\end{proof}

\section{Details of Numerical Example}
Consider the Ising Hamiltonian under the homogeneous magnetic field for two spins, $\beta H(\sigma_1, \sigma_2) = - \beta J \sigma_1 \sigma_2 - \beta h_{\mathrm{ex}} (\sigma_1 + \sigma_2)$. We consider the situation that the agent changes the magnetic field with fixed coupling constant. We set $\beta J = 1$ for simplicity and write $\theta \coloneqq \beta h_{\mathrm{ex}}$. The initial parameter is set to be $\theta = 0$. The average magnetization $\mu (\theta) = \mathbb{E}_{\theta} [ \sigma_1 + \sigma_2 ]$ is given by $\mu (\theta) = 2(e^{2 \theta} + e^{- 2 \theta})/(e^{2\theta} + e^{-2 \theta} + 2 e^{-2})$. The maximum likelihood estimator for $\theta$ with respect to a compact set $\Theta_0 = [ - \ln 2, \ln 2]$ is given by
\begin{align}
 \hat{\theta}_{\mathrm{ML}} (\omega^n) = 
 \begin{cases}
 - \ln 2 & \text{if } \bar{\phi}_n < \mu ( - \ln 2)
 \\
 \theta ( \bar{\phi}_n ) & \text{if } \bar{\phi}_n \in [ \mu ( - \ln 2), \mu ( \ln 2) ]
  \\
 \ln 2 & \text{if } \bar{\phi}_n > \mu ( \ln 2),
 \end{cases}
\end{align}
where
\begin{align}
 \theta (\mu ) = \frac{1}{2} \ln \left( \frac{\mu e^{-2} + \sqrt{4 + (e^{-4} -1) \mu^2}}{2-\mu} \right)
\end{align}
is the inverse of $\theta \mapsto \mu (\theta)$. If a sequence is generated by the Gibbs distribution with $\theta \neq 0$, the empirical mean for the sequence converges to $\mu (\theta)$ but it is different from $\mu (0)$. Fig. \ref{fig:work} shows that the accumulation of extracted work under the maximum likelihood strategy diverges to infinity for such sequences.

\end{document}